\documentclass[10pt,conference]{IEEEtran}
\usepackage[cmex10]{amsmath}
\usepackage{amssymb}
\usepackage{amsthm}
\usepackage{epsfig}
\usepackage{color}
\usepackage{algorithm}
\usepackage{algorithmic}
\usepackage{multirow}
\usepackage{array}
\newtheorem{theorem}{Theorem}

\newtheorem{corollary}{Corollary}
\newtheorem{lemma}{Lemma}
\usepackage{epsfig}
\usepackage{color}

\begin{document}
\title{Data Attacks on Power Grids: Leveraging Detection}
\author{\authorblockN{Deepjyoti Deka,~ Ross Baldick ~and~ Sriram Vishwanath}
\authorblockA{Department of Electrical \& Computer Engineering, The University of Texas at Austin\\
Email: deepjyotideka@utexas.edu, baldick@ece.utexas.edu, sriram@ece.utexas.edu }}

\maketitle
\begin{abstract}
Data attacks on meter measurements in the power grid can lead to errors in state estimation. This paper presents a new data attack model where an adversary produces changes in state estimation despite failing bad-data detection checks. The adversary achieves its objective by making the estimator incorrectly identify correct measurements as bad data. The proposed attack regime's significance lies in reducing the minimum sizes of successful attacks to more than half of that of undetectable data attacks. Additionally, the attack model is able to construct attacks on systems that are resilient to undetectable attacks. The conditions governing a successful data attack of the proposed model are presented along with guarantees on its performance. The complexity of constructing an optimal attack is discussed and two polynomial time approximate algorithms for attack vector construction are developed. The performance of the proposed algorithms and efficacy of the hidden attack model are demonstrated through simulations on IEEE test systems.
\end{abstract}

\section{Introduction}
One of the basic facets of research and actual deployment in smart grid has been increased data collection from different meters for improved monitoring and control of dynamic events. Accurate data collection also aids formation of optimal prices and price-responsive demand. However, a data driven approach makes the grid vulnerable to cyber-attacks on meter measurements. A coordinated data attack on meter recordings in Phasor Measurement Units (PMUs) \cite{pmu1} and Remote Terminal Units (RTUs) or on the communication channels in Supervisory Control and Data Acquisition (SCADA) systems can in principle lead to incorrect electricity prices as well as to large blackouts.

Data attacks on meter measurements in the power grid is an active area of research. The authors of \cite{hidden} first introduced the problem of undetectable data attacks that bypass bad-data tests at the state estimator. Simple linear algebraic techniques show that if the malicious measurements lie in the column space of the measurement matrix, the attack goes undetected. In reference \cite{hidden}, an attack vector consisting of the malicious measurements is constructed using projection matrices based on the measurement matrix. This work has been followed by several techniques to select locations for introducing data attacks under different grid conditions and adversarial objectives. Reference \cite{poor} discusses the construction of an optimal hidden attack that requires manipulation of the minimum number of measurements using $l_0$ and $l_1$ recovery methods. Reference \cite{sou} studies the creation of the optimal attack vector as a mixed integer linear program. The authors of \cite{deka} discuss graph based design of optimal attack vectors for systems observed by PMUs. Data attacks aimed at affecting the estimation of a pre-specified set of state variables is presented in \cite{deka1}. A heuristic based detector for malicious data is presented in \cite{thomas}.

A majority of the prior work in this area focus on constructing hidden data attacks that evade bad-data detection tests at the state estimator. In this paper, we analyze a detectable regime of data attacks. It is worth mentioning that at the state estimator, bad data detection is followed by a scheme for bad-data identification. Our proposed attack regime succeeds despite detection by deceiving the bad-data identifier into labeling uncorrupted measurements as bad data. While writing this manuscript, we discovered a related work demonstration detectable data attacks in \cite{frame}. The attack model in \cite{frame} searches for an optimal hidden attack and then creates a detectable attack by corrupting only half of the measurements necessary for a hidden attack. In a sense, our model is a generalization of the framework used in \cite{frame}, with some key differences. Our attack model considers power system cases where a subset of the measurements are incorruptible by the adversary. By overcoming the presence of incorruptible measurements, the cardinality of our data attacks can be reduced by more than $50\%$ of the cardinality of the optimal hidden attack, whereas in \cite{frame}, the reduction in cardinality is exactly by half. More importantly, unlike the framework in \cite{frame} our detectable attack model is able to greatly expand the range of feasible attacks to configurations where no hidden attacks are possible. Further, we show that considering incorruptible measurements in the system makes the problem of constructing an optimal detectable attack NP-hard in general. In contrast, detectable attacks that do not overcome incorruptible measurements can be constructed in polynomial time in our measurement set-up.

The rest of this paper is organized as follows. The next section presents a description of the system model used in state estimation and bad-data detection, and then introduces our data attack model. We derive conditions necessary for a successful attack of our regime and provide provable guarantees on their cardinality in Section \ref{sec:adversary solution}. The two approximate algorithms to design an optimal attack vector for our regime are presented in Section \ref{sec:algo}. Both of them require information limited to the structure of the measurement matrix and do not need the numerical values of grid parameters. Simulations of the proposed algorithms on test IEEE bus systems are shown in Section \ref{sec:results}. Finally, concluding remarks and future directions of work are presented in Section \ref{sec:conclusion}.

\section{Power Grid State Estimation and Attack Models}
\label{sec:attack}
We represent the power grid by an undirected graph $(V,E)$, where $V$ denotes the set of buses and $E$ represents the set of transmission lines connecting those buses. In this paper, we consider DC power flow model \cite{abur} for state estimation in the grid that is given by:
\begin{align}
z = Hx + e \label{dcmodel}
\end{align}
Here $z \in \mathbb{R}^m$ is the $m$ length vector of measurements. We consider two kinds of measurements collected through conventional meters and PMUs in the grid. These include: a) flow measurements on lines and b) voltage phasor measurements on buses. $x \in \mathbb{R}^n$ is the state vector of length $n$ and consists of the bus phase angles. $H$ is the measurement matrix and $e$ is a zero mean Gaussian noise vector with covariance $\Sigma$. Let the $k^{th}$ entry in $z$ measure the power flow on the line between nodes $i$ and $j$. We have $z(k) = B_{ij}x(i)-B_{ij}x(j)$, where $B_{ij}$ is the magnitude of susceptance of the line $(i,j)$. The corresponding $k^{th}$ row in the measurement matrix $H$ is given by $H_{k} = [0..0~~B_{ij}~~ 0..0~~-B_{ij}~~0..0]$. Similarly, let the $l^{th}$ entry in $z$ measure the phase angle at bus $i$. The corresponding row $H_l$ in the measurement matrix is $H_{l} = [0..0 ~~ 1~~ 0..0]$ with one at the $i^{th}$ position. We assume $m>n$ and full column rank of matrix $H$, as necessary for unique state estimation.

\textbf{State Estimator:} The schematic diagram of the state estimator in the grid is shown in Figure \ref{estimator} \cite{monticelli,abur}.
\begin{figure}
\centering
\includegraphics[width=0.44\textwidth]{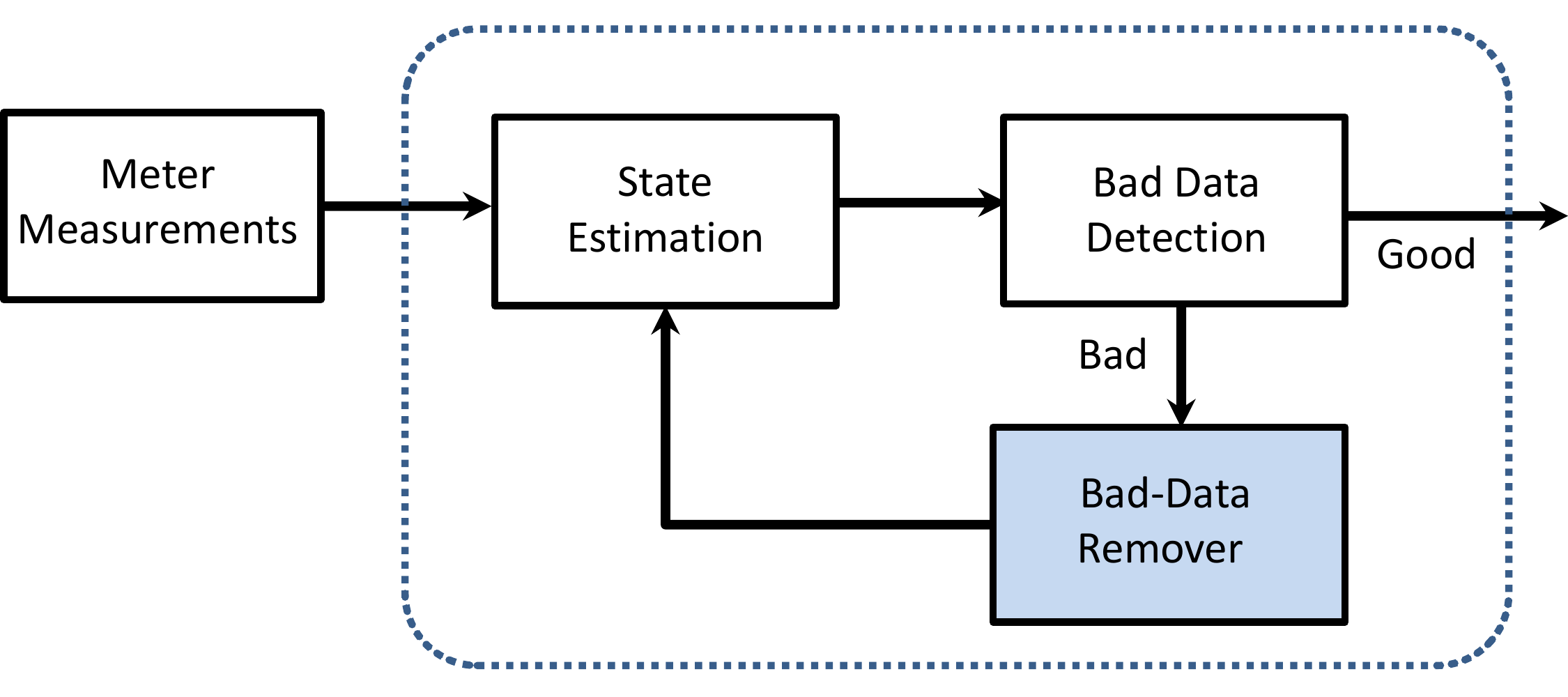}
\caption{State Estimator for a power system \cite{monticelli,abur}}
\label{estimator}
\end{figure}
For the DC model in Equation (\ref{dcmodel}), a Weighed Least Square (WLS) Estimator gives the optimal state vector estimate $x^*$ by minimizing the weighted measurement residual's magnitude $J(x,z) = \|\Sigma ^{-.5}(z-Hx)\|_2$. The estimator then uses the following threshold ($\lambda$) based test to detect the presence of bad-data.
\begin{align}
\|\Sigma ^{-.5}(z-Hx^*)\|_2 &\leq \lambda ~~\text{accept~~~} x^*\nonumber\\
                            &> \lambda ~~\text{detect bad-data} \label{test}
\end{align}

\textbf{Removal of Bad-data:} Once bad-data is detected, the estimator tries to remove the bad-data and then re-estimate the state vector. The measurement residual vector $r$ corresponding to the estimated $x^*$ is given by $r = z-Hx^* = [I - H(H^T\Sigma^{-1}H)^{-1}H^T\Sigma^{-1}]z$ \cite{monticelli,abur}. Using this relation, we can derive the variance $R_r$ of residual $r$ as well as the normalized residual of the data.
It is shown in \cite{monticelli,abur} that for the general case of multiple bad data entries (as in our case), a sequential bad-data remover described in previous literature, is sub-optimal. The optimal strategy for the estimator is to \textit{remove the minimum number of measurements} such that the residual produced by the remaining measurements passes the bad-data detection test in (\ref{test}). In addition, the estimator needs to ensure that the removed measurements do not lead to a loss of rank in the measurement matrix as that will make the system unobservable. The optimal bad-data removal procedure is formulated as the following non-convex problem \cite{monticelli}:
\begin{align}
&\min_{d \in \{0,1\}^m} \|(\textbf{1} - d)\|_0 \label{ident}\\
\text{s.t. ~} & rank (H_d) = n,~~ J(x*, z_d) \leq \lambda_d
\end{align}
Here, $H_d$, $z_d$, $J(x*, z_d)$ and $\lambda_d$ respectively denote the updated measurement matrix, measurement vector, minimum weighted residual magnitude and threshold obtained after the measurements corresponding to $0$ entries in $d$ are removed.

\subsection{Attack Models}
 Let the adversary introduce an attack vector $a$ in the measurements to generate the corrupted measurement vector $\hat{z} = z+a$. We assume that the adversary is interested in constructing a feasible attack using minimum number of corrupted measurements ($\|a\|_0$). In a realistic setting, an adversary may be incapable of modifying certain measurements due to geographical isolation or heightened encryption. We call this set of incorruptible measurements as $S_m$ and the complimentary set of corruptible measurements as $S^c_m$. Note that measurements in $S_m$ suffer from noise and measurement errors; they are just free of adversarial manipulation. Next, we briefly describe hidden data attacks that bypass bad-data detection checks.

\textbf{Undetectable Data Attack:} Observe that if $a = Hc$, the measurement residual stays the same as $\|\Sigma ^{-.5}(z-Hx^*)\|_2 = \|\Sigma ^{-.5}(z+a-H(x^*+c))\|_2$. Thus, an erroneous state vector $x^* +c$ is produced without raising any alarm at the bad-data detector \cite{hidden}. The solution to Problem \ref{opt_attack} below gives the adversary's optimal attack vector \cite{poor,deka}.
\begin{align} \label{opt_attack} \tag{P-1}
&\min_{c \in \mathbb{R}^n-\{\textbf{0}\}} \|a\|_{0} \\
\text{s.t. ~} &a =Hc, ~a(i) = 0 ~\forall i \in S_{m} ~~(\text{$S_m$: incorruptible set}) \nonumber
\end{align}
\textbf{Data Attack with Detection:} We now discuss our proposed detectable attack model. We assume that without any adversarial manipulation, measurement $z$ or any observability preserving subset of $z$ is capable of producing a correct state estimate $x^*$. Consider a data attack vector $a$ that fails the bad-data detection test. For the bad-data identification scheme given in (\ref{ident}), this data attack can nonetheless change the state estimate if \emph{removal of $k < \|a\|_0$ measurements is sufficient to satisfy the bad-data detection test while maintaining system observability.}
This provides the conditions needed by a feasible dat attack of our proposed model. Construction of an attack vector for this regime is given by the following optimization problem:
\begin{align} \label{opt_attacknew} \tag{P-2}
&\min_{d \in \{0,1\}^m } \|a\|_{0} \\
\text{s.t. ~} &a = d*(Hc), c \in \mathbb{R}^n-\{\textbf{0}\}\\\nonumber
&a(i) = 0 ~\forall i \in S_{m} ~~(\text{$S_m$: incorruptible measurements}) \nonumber\\
& \|a\|_0 > \|(\textbf{1}-d)*(Hc)\|_0 \label{cond1} \\
& rank(DH) = n \text{~~where~} diag(D) = d \label{cond2}
\end{align}
In Problem \ref{opt_attacknew}, $D$ is a diagonal matrix with diagonal given by vector $d$. $a*b$ represents element-wise multiplication between two vectors $a$ and $b$. Unlike Problem \ref{opt_attack}, here $a$ does not lie in the column space of $H$ as certain entries in the attack vector are deleted by the binary vector $d$. Condition (\ref{cond1}) ensures that the estimator incorrectly identifies uncorrupted elements of $z$ as bad-data. Note that after removal of bad data, the attack vector $a = d*(Hc)$ passes the bad-data detection test as it lies in the column space of the updated measurement matrix $DH$. In the next section, we discuss the design of an optimal attack vector for Problem \ref{opt_attacknew}. 

\section{Optimal Attack Vector Design}
\label{sec:adversary solution}
Consider the DC measurement model for a $n$ bus system given in Equation (\ref{dcmodel}). We now introduce a $(n+1)^{th}$ reference bus with phase angle $0$ and augment $c$ to form vector $\hat{c} = \setlength{\arraycolsep}{2pt} \renewcommand{\arraystretch}{0.8}\begin{bmatrix} c \\0 \end{bmatrix}$. We also add one extra column $h^{g}$ after the rightmost column in measurement matrix $H$ to create a $m$ times $(n+1)$ modified measurement matrix $\hat{H}$. We put $-1$ in $h^g$ for every row in $H$ with a phase angle measurement and $0$ otherwise. We now have
$Hc = \hat{H}\hat{c} = \left[H ~|~ h^g\right]\begin{bmatrix} c \\0 \end{bmatrix}$. Note that all rows in the augmented measurement matrix $\hat{H}$ represent flows. We now state the following theorem without proof from \cite{deka}.
\begin{theorem}[{\cite[Theorem 1]{deka}}] \label{01}
There exists a non-zero binary $0-1$ vector $c_{opt}$ of size $n$ times $1$ for an optimal attack vector $a^*$ for Problem \ref{opt_attack} such that $\|a^*\|_0 = \|Hc_{opt}\|_0$.
\end{theorem}

In a similar way, it can be proven through contradiction that \textbf{the optimal attack vector $a^*_1$ for Problem \ref{opt_attacknew} also corresponds to a non-zero binary $0-1$ vector $c_{opt1}$ with $\|a^*\|_0 = \|d*Hc_{opt1}\|_0$}.

Next we create a new matrix $A_H$ by replacing magnitudes of all bus susceptance in $\hat{H}$ with unity.
\begin{align}
A_H(i,j) = 1(\hat{H}(i,j) > 0) - 1(\hat{H}(i,j) < 0) \label{incident}
\end{align}
Observe that $A_H$ represents the incident matrix for a graph with $n+1$ nodes, with edges corresponding to measurements in $\hat{H}$. We denote the graph represented by $A_H$ as $G_H$. The $(n+1)^{th}$ node in $G_H$ represents the reference bus with phase angle $0$. Notice that for any $0-1$ vector $\hat{c} = \setlength{\arraycolsep}{2pt} \renewcommand{\arraystretch}{0.8}\begin{bmatrix} c \\0 \end{bmatrix}$, $\|\hat{H}\hat{c}\|_0 = \|A_H\hat{c}\|_0$, where the non-zero values of $A_H\hat{c}$ represent a cut in graph $G_H$ between the nodes with $\hat{c}(i) = 0$ and the nodes with $\hat{c}(i) = 1$. 
We now write the attack vector design for our proposed Problem \ref{opt_attacknew} in terms of $A_H$ as:

\begin{align} \label{opt_attacknew1} \tag{P-3}
&\min_{d \in \{0,1\}^m } \|a\|_{0} \\
\text{s.t. ~} &a = d*(A_H\hat{c}), a(i) = 0 ~\forall i \in S_{m} \nonumber\\
&\hat{c} \in \{0,1\}^{n+1}- \{\textbf{0}\},~ \hat{c}(n+1) = 0\nonumber\\
& \|d*(A_H\hat{c})\|_0 > \|(\textbf{1}-d)*(A_H\hat{c})\|_0 \label{cond1a}\\
& rank(DA_H) = n \text{~~where~} diag(D) = d \label{cond2a}
\end{align}
Observe that non-zero values in $A_H\hat{c}$ define a graph-cut in $G_H$, out of which the edges with value $1$ in $d$ are included in the attack vector $a$. $a$ of course does not include any edge in $S_m$. Further, condition (\ref{cond1a}) implies that an attack vector is feasible if the number of cut-edges included in the attack vector $a$ is strictly greater than half of the cut-size. Our principal result on constructing an optimal data attack of our proposed regime is given in the following theorem.

\begin{theorem}\label{main}
Let $C_{a^*}$ be a minimum cardinality cut in $G_H$ such that the number of cut-edges in $C_{a^*}$ that belong to $S_m$ is strictly less than half of the cut-size $|C_{a^*}|$. An optimal attack vector for Problem \ref{opt_attacknew1} is given by a sub-set of cut-edges in $C_{a^*} \cap S^c_m$ of cardinality $\lfloor1+ |C_{a^*}|/2\rfloor$.
\end{theorem}
\begin{proof}
Let $a^*$ denote the attack vector with non-zero entries corresponding to $\lfloor1+ |C_{a^*}|/2\rfloor$ edges in $C_{a^*} \cap S^c_m$. Thus $\|a^*\|_0$ is greater than $|C_{a^*}|/2$ and condition (\ref{cond1a}) is satisfied. The edges of $C_{a^*}$ excluded from $a^*$ are removed as bad-data by the estimator. System observability is preserved if graph $G_H$ stays connected after bad-data removal. If the graph becomes disconnected after the bad-data removal, we can form a smaller feasible graph-cut using a subset of the removed edges and the ones with non-zero values in $a^*$. This  contradicts the definition of $C_{a^*}$. Hence observability is maintained by $C_{a^*}$.
\end{proof}

We now prove some important results on the adversarial potential of data attacks of our regime as compared to undetectable attacks.

\begin{lemma}\label{main1}
Let $a^*_u$ (undetectable attack) and $a^*_d$ (detectable attack) respectively be the optimal attack vector designs for Problem \ref{opt_attack} and Problem \ref{opt_attacknew}, formulated for the same system. Then the following holds: $\|a^*_d\|_0 \leq \lfloor(1+ \frac{\|a^*_u\|_0}{2})\rfloor$
\end{lemma}
\begin{proof}
Note that if we fix $d = \textbf{1}$ in Problems \ref{opt_attacknew} ($\equiv$ Problem \ref{opt_attacknew1}), it reduces to Problem \ref{opt_attack}. The optimal undetectable attack vector $a^*_u$ is given by $Hc_u ~ \exists c_u \neq \textbf{0}$. 
Let $\|a^*_u\|_0 = k$ with the non-zero entries in $a^*_u$ being located at positions $1$ to $k$. Consider $a_d$ such that $a_d(i) = a^*_u(i)$ for $i \in \{1,\lfloor(1+ \frac{k}{2})\rfloor\}$ and $0$ elsewhere. It can be easily verified that $a_d$ is a feasible detectable attack for Problem \ref{opt_attacknew1}. As $a^*_d$ is the optimal attack for Problem \ref{opt_attacknew1}, we have $\|a^*_d\|_0 \leq \|a^d\|_0 = \lfloor(1+ \frac{\|a^*_u\|_0}{2})\rfloor$.
\end{proof}

Note that this provides only an upper bound on the cardinality of optimal attack vectors and in practice, the reduction in cardinality can be much greater. Further the following is true:

\begin{lemma}\label{main2}
The set of system conditions with feasible data attacks for Problem \ref{opt_attacknew1} is larger than that for Problem \ref{opt_attack}.
\end{lemma}
Observe that every undetectable attack can give a corresponding feasible detectable data attack with detection for our proposed regime. However, data attacks with detection can exist for cases where no undetectable attacks are possible. For example, no undetectable attack exists if every cut in $G_H$ includes at least one incorruptible measurement. An attack with detection may exist if the number of edges of $S_m$ in a cut is less than the cut-size. The next result (proof omitted for brevity) shows that for preventing detectable attacks, at least $50\%$ of the measurements need to be incorruptible.

\begin{corollary}\label{main3}
An attack vector for Problem \ref{opt_attacknew1} always exists if less than half of the measurements in the system are incorruptible.
\end{corollary}
Thus, the number of incorruptible measurements needed to prevent a detectable attack scales with the number of edges, whereas that needed to prevent hidden attacks scales with the number of nodes in $G_H$. In the next section, we discuss two approximate algorithms to generate attack vectors for this regime.

\section{Algorithm For Attack Vector Construction}
\label{sec:algo}
Theorem \ref{hardness} gives the computational complexity of finding the optimal attack vector in Problem \ref{opt_attacknew1}.
\begin{theorem}\label{hardness}
Problem \ref{opt_attacknew1} is NP-hard.
\end{theorem}
\textbf{Proof steps:}
We prove this by showig the NP-hardness of determining the existence of a feasible solution of Problem \ref{opt_attacknew1}. Consider the case of $G_H$ being a complete graph with cut $C$ separating the nodes into sets $A$ and $A^c$. Let the number of edges of $S_m$ in cut $C$ be $cut_{S_m}(A,A^c)$. For feasible attack, $cut_{S_m}(A,A^c)$ should be less than half of the cut-size $|A||A^c|$ (complete graph). Thus, we need $\frac{cut_{S_m}(A,A^c)}{|A||A^c|} < .5 $. This represents a ratio-cut of value less than $.5$. Thus our problem is equivalent to establishing the existence of a ratio-cut of value less than $.5$, a known NP-complete problem \cite{ratio}.

We now discuss two approximate schemes to find the optimal solution for Problem \ref{opt_attacknew1}. The first scheme uses a Semi-definite Programming based randomized approach \cite{SDP}.\\
\textbf{(a) SDP approach:} Our SDP based technique builds upon the randomized solution for max-cut given by Goemans and Williamson \cite{SDP}. The following is a SDP relaxation to find the optimal feasible cut in $G_H$.\\


\begin{tabular}{ccc}
$\begin{aligned}
\min_{x \in \{-1,1\}^{n+1}} \langle L^1_H, xx^T\rangle& \nonumber\\
\text{s.t. ~}  
\frac{\langle L^2_H, xx^T \rangle}{4} \leq  -1&
\end{aligned}$
&$ \xrightarrow[xx^T \text{to} X]{\text{Relax}}$&
$\begin{aligned} \label{SDP1}
\min_{X \in S^{n+1}} \langle L^1_H, X\rangle& \text{~~~(P-4)} \\
\text{s.t. ~}  diag(X) =& \textbf{1}\nonumber\\
\langle L^2_H, X \rangle \leq  -&4 
\end{aligned}$
\end{tabular}\\

Here, $S^{n+1}$ is the space of positive semidefinite matrices of size $(n+1)$. $L^1_H$ is the standard Laplacian matrix for graph $G_H$ with edge-weights unity, while, $L^2_H$ is a modified Laplacian for $G_H$ where edges in $S_m$ and $S^c_m$ are given weights of $1$ and $-1$ respectively. The original problem tries to label nodes in $G_H$ with values in $\{-1,1\}$ so that $\langle L^1_H, xx^T\rangle/2$ represents the cut-size. $\frac{\langle L^2_H, xx^T \rangle}{4} \leq  -1$ ensures that the cut contains greater number of edges of $S^c_m$ than of $S_m$. Following the work in \cite{SDP}, we give randomized Algorithm $1$ for Problem \ref{opt_attacknew1}.
\begin{algorithm}
\caption{SDP Relaxation for Problem \ref{opt_attacknew1}}
\begin{algorithmic}[1]
\STATE Solve Problem P-4 to get $X^*$. Generate $X^* = B^TB$ by Cholesky decomposition.
\STATE Randomly pick a vector $w \in \mathbb{R}^{n+1}$.
\FOR {$i = 1 \TO n+1$}
  \STATE $x(i) = 1(\langle B(i,:), w\rangle \geq 0) - 1(\langle B(i,:), w\rangle < 0)$
\ENDFOR
\IF {$\langle L^2_H, xx^T \rangle< 0$}
   \STATE Output optimal attack as subset of $1+ \lfloor\frac{\text{cut-size}}{2}\rfloor$ edges of $S^c_m$ in cut defined by $x$.
\ENDIF
\end{algorithmic}
\end{algorithm}

\textbf{(b) Iterative Min-cut approach:} The min-cut of $G_H$ computed using unit-edge weights may contain more edges of set $S_m$ than of $S^c_m$ and can be infeasible. Our approximate Algorithm $2$ tries to overcome this by iteratively computing the min-cut with smaller edge-weights for set $S^c_m$. In particular, steps \ref{step1} to \ref{step8} reduce the edge-weights for $S^c_m$ to replace the current infeasible min-cut ($C$) of cardinality $c$ by a feasible cut (if it exists) of cardinality $c+b$. If lowering edge weights alone is not sufficient to achieve a feasible cut, step \ref{step9} chooses one edge in $C$ randomly that belongs to $S_m$ and gives it infinite weight. The algorithm iterates until a feasible solution is found or all edges of $S_m$ have been given infinite weights, which indicates absence of any solution.
\begin{algorithm}
\caption{Iterative Min-cuts for Problem \ref{opt_attacknew1}}
\begin{algorithmic}[1]
\STATE Give edge weight $1- \epsilon$ in $S^c_m$. Compute min-cut $C$ in $G_H$.
\STATE $c \gets |C|, c_m \gets \sum_{i\in C} 1(i \in S_m), b \gets 1$
\WHILE {($c < \infty, 2c_m \geq c$)}
\IF {$2c_m \geq b + c$}  \label{step1}
   \STATE $\beta \gets 1 - \epsilon - b(c_m + b - \lfloor (c +b- 1)/2\rfloor)^{-1}$ \label{step2}
   \STATE Give edge weight $\beta$ in $S^c_m$, get min-cut $C_1$ in $G_H$. \label{step3}
   \IF {$|C| = |C_1|$} \label{step4}
      \STATE $b \gets b+1$ \label{step5}
   \ELSE \label{step6}
      \STATE $C \gets C_1, c \gets |C|, c_m \gets \sum_{i\in C} 1(i \in S_m), b \gets 1$ \label{step7}
   \ENDIF \label{step8}
\ELSE
   \STATE Randomly pick edge $i \in C \cap S_m$, give $\infty$ weight. \label{step9}
   \STATE $\beta \gets 1 -\epsilon$. Give weight of $\beta$ to s $S^c_m$. Compute min-cut $C$ in $G_H$.
   \STATE $c \gets |C|$, $c_m \gets \sum_{i\in C} 1(i \in S_m), b \gets 1$.
\ENDIF
\ENDWHILE
\IF {$|C| \neq \infty$}
   \STATE Output optimal attack as $1+ \lfloor\frac{c}{2}\rfloor$ edges of $S^c_m$ in $C$.
\ENDIF
\end{algorithmic}
\end{algorithm}

The theoretical analysis of the exact expressions for run-time and performance guarantees of the proposed algorithms will be covered in a future work. In the next section, we show simulation results on the performance of the algorithms on IEEE  test systems.

\section{Results on IEEE test systems}
\label{sec:results}
 We run simulations in Matlab Version 2009a and present averaged results in this section. We consider the IEEE $14$ bus test system \cite{testsystem}. In this system, we put phasor measurements on $60\%$ of the buses and flow measurements on all lines. Figure \ref{attacksize} shows the average sizes of the best attack vectors constructed by our proposed solution schemes discussed in the previous section. As expected, the size of the attack vector increases with increase in the number of incorruptible measurements ($S_m$) in the system. Moreover, we plot the average size of undetectable attacks in the same figure to show the significant improvement in cardinality offered by our attack regime; the improvement being greater for Algorithm $2$ than Algorithm $1$. Next, Figure \ref{attackno} plots the rise in the average fraction of cases resilient to data attacks with number of incorruptible measurements (set $S_m$). Observe that, unlike undetectable attacks that increasingly become infeasible at higher levels of secure measurements, feasible attacks of our proposed model are still designed by Algorithm $1$ and $2$. Thus, total resilience against attacks of our regime requires greater placement of secure measurements than that needed for protection against undetectable attacks. Both these figures validate our claim that data attacks with detection are far more potent than previously studied undetectable attacks.

\begin{figure}[ht]
\centering
\includegraphics[width=0.44\textwidth]{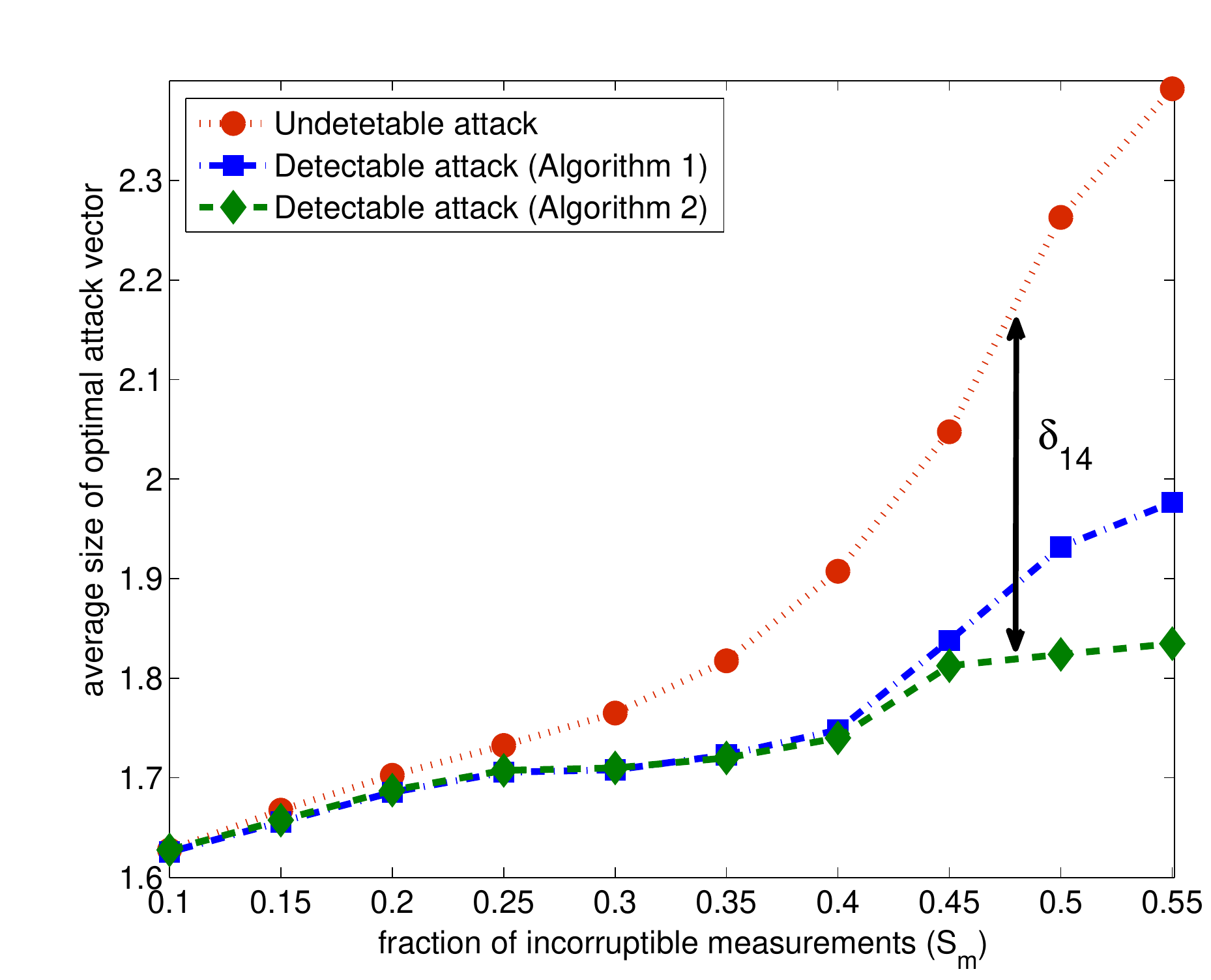}
\caption{Average size of feasible attacks given by Algorithm $1$ and $2$ for IEEE $14$ bus test system with flow measurements on all lines, phasor measurements on $60\%$ of the buses and protection on a fraction of measurements selected randomly. $\delta_{14} > 1 + \lfloor c/2\rfloor$ where $c$ is size of undetectable attack}
\label{attacksize}
\end{figure}
\begin{figure}
\centering
\includegraphics[width=0.41\textwidth]{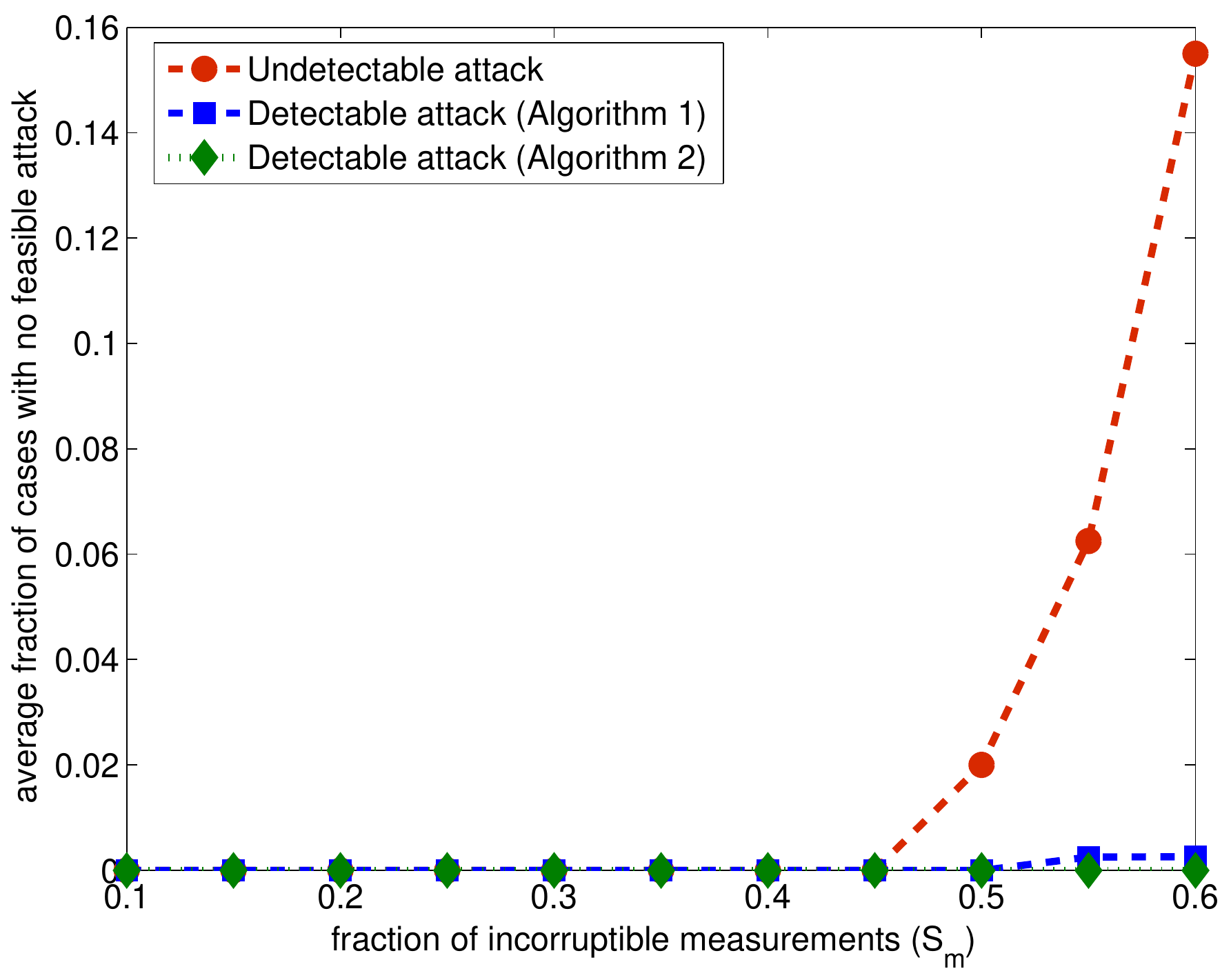}
\caption{Average fraction of simulated test cases with no feasible attacks given by Algorithm $1$ and $2$ for IEEE $14$ bus test system with flow measurements on all lines, phasor measurements on $60\%$ of the buses and protection on a fraction of measurements selected randomly.}
\label{attackno}
\end{figure}

\section{Conclusion}
\label{sec:conclusion}
We propose a new framework of detectable data attacks on state estimation that operate by making the state estimator incorrectly label and remove good measurements as bad-data. The minimum number of measurements that need to be manipulated for a successful detectable data attack is upper bounded by half of that needed for previously studied undetectable data attacks. We show that the optimal attack of our regime is given by the minimum cardinality graph cut satisfying a feasibility constraint. We prove that the problem of designing the optimal detectable attack is NP-hard and present two approximate algorithms for it. Simulations of attack vector construction on IEEE $14$-bus system demonstrate that our attack regime undermines the security of state estimation further than current attack models. We are currently studying guarantees on the performance of our algorithms and design of protection schemes against our attack framework.

\end{document}